\documentclass[12pt,reqno]{amsart}
\usepackage{amsmath}
\usepackage{latexsym}
\usepackage{amsfonts}
\usepackage{amssymb}
\usepackage{color}
\usepackage{bbm,dsfont}
\usepackage{graphicx}
\usepackage{enumerate}

\newtheorem{theorem}{Theorem}

\newtheorem{proposition}{Proposition}

\theoremstyle{definition}
\newtheorem{example}{Example}

\newcommand{\Rb}{\mathbb R} 
\newcommand{\Cb}{\mathbb C} 
\newcommand{\half}{\tfrac{1}{2}} 

\newcommand{\hh}{\mathcal{H}} 
\newcommand{\kk}{\mathcal{K}} 
\newcommand{\lh}{\mathcal{L(H)}} 
\newcommand{\lk}{\mathcal{L(K)}} 
\newcommand{\ip}[2]{\left\langle\,#1\,|\,#2\,\right\rangle} 
\newcommand{\ket}[1]{|#1\rangle} 
\newcommand{\kb}[2]{|#1\rangle\langle#2|} 
\newcommand{\no}[1]{\left\|#1\right\|} 
\newcommand{\tr}[1]{\textrm{tr}\left[#1\right]} 
\newcommand{\id}{\mathbbm{1}} 

\newcommand{\va}{\mathbf{a}} 
\newcommand{\vsigma}{\boldsymbol{\sigma}} 

\newcommand{\obs}{\mathcal{O}}
\newcommand{\A}{\mathsf{A}}
\newcommand{\B}{\mathsf{B}}
\newcommand{\C}{\mathsf{C}}
\newcommand{\Eo}{\mathsf{E}}
\newcommand{\F}{\mathsf{F}}
\newcommand{\G}{\mathsf{G}}

\newcommand{\chan}{\mathcal{C}}

\newcommand{\Ii}{\mathcal{I}}

\newcommand{\csigma}{\sigma_{\mathrm{c}}} 
\newcommand{\ctau}{\tau_{\mathrm{c}}}

\newcommand{\clo}{\mathfrak{c}}
\newcommand{\simu}{\mathfrak{sim}} 
\newcommand{\leak}{\mathfrak{leak}} 
\newcommand{\joint}{\mathfrak{joint}} 

\newcommand{\comp}{\hbox{\hskip0.85mm$\circ\hskip-1.2mm\circ$\hskip0.85mm}}

\makeatletter
\def\@settitle{\begin{center}%
  \baselineskip14\p@\relax
  \bfseries
  \uppercasenonmath\@title
  \@title
  \ifx\@subtitle\@empty\else
     \\[1ex]\@subtitle
     
  \fi
  \end{center}%
}
\def\subtitle#1{\gdef\@subtitle{#1}}
\def\@subtitle{}
\makeatother

\begin{document}

\title[]{Noise-disturbance relation and the Galois connection of quantum measurements}
\subtitle{\emph{Dedicated to the memory of Paul Busch}}

\author[]{Claudio Carmeli \and Teiko Heinosaari \and Takayuki Miyadera \and Alessandro Toigo
}

\maketitle

\begin{abstract}
The relation between noise and disturbance is investigated within the general framework of Galois connections.
Within this framework, we introduce the notion of leak of information, mathematically defined as one of the two closure maps arising from the observable-channel compatibility relation. We provide a physical interpretation for it, and we give a comparison with the analogous closure maps associated with joint measurability and simulability for quantum observables.
\end{abstract}

\section{Introduction}
A fundamental fact about quantum measurements is the following: measurement that does not cause any disturbance cannot give any information on the measured system. 
One of the most compact and instructive proofs of this fact, using only the basics of functional analysis, was presented by Paul Busch in \cite{Busch09}.
This no-go theorem motivates for further investigation, namely, to analyze what kind of noise must be tolerated for certain kind of disturbance, and vice versa, what is the minimal possible disturbance if certain noise is accepted. 
The aim of this paper is to provide some insight into one aspect of this general question.

A simplified but useful framework to think of measurements is to consider them as devices that have an input port for the measured system and two different ports for the output, one that gives the measurement outcome distribution and the other one that gives the transformed state. 
If we only consider the measurement outcomes we have an observable, while considering only the transformed state yields a channel.
A quantum observable and a quantum channel are called compatible if they are parts of a single measurement device, otherwise they are incompatible. 
In this language, the no-information-without-disturbance theorem states that the identity channel is compatible only with coin tossing observables.

The qualitative noise-disturbance relation, presented in \cite{HeMi13} and further developed in \cite{HeMi15,HeMi17,HaHeMi18}, characterizes the compatible channels for any given observable: the set of compatible channels is a principal ideal, generated by the so-called least disturbing channel of that observable.
We would like to point out that the work that led to \cite{HeMi13} started when Paul recommended two of the authors, not known to each other before, to meet for a scientific interaction.
Paul's encouragement, advice and support were important for that work, as they were for many of our works before and after that. 

The qualitative noise-disturbance relation leads to the following conclusion: if we know all compatible channels of an unknown observable, then we can recover that observable up to post-processing equivalence.
Therefore, a natural generalization of the qualititative noise-disturbance relation is to consider the set of all compatible channels for a collection of observables, instead of a single observable. 
The mathematical framework to investigate this correspondence is the Galois connection induced by the compatibility relation. 
Forming the Galois connection gives immediately two closure maps, one on the set of observables and another one on the set of channels. 
The physical interpretation of the maps involved in the Galois connection is not anymore as direct as in the qualitative noise-disturbance relation.
We will explain how the closure map on the set of observables gives a mathematical description of information leak. 

This paper is organized as follows.
In Section \ref{sec:qndr} we recall the qualitative noise-disturbance relation and some other background concepts and results. 
In Section \ref{sec:galois} we formulate the Galois connection of observables and channels and derive some of its properties.
The physical interpretation of one of the resulting closure maps is explained in Section \ref{sec:leak}. 
Finally, in Section \ref{sec:joint} we form another Galois connection and compare the resulting closure map with the previously obtained closure map.

\section{Qualitative noise-disturbance relation}\label{sec:qndr}

\subsection{Preliminaries and notations}

In the following, we always deal with finite dimensional quantum systems. We fix one of such systems, and denote by $\hh$ its associated Hilbert space. We let $\lh$ be the linear space of all complex linear operators on $\hh$, and write $\id$ for the identity operator.

An {\em observable} with outcomes in a finite set $\Omega$ is a map $\A:\Omega\to\lh$ such that $\A(\omega)$ is a positive operator for all $\omega\in\Omega$, and $\sum_\omega \A(\omega) = \id$. A {\em channel} with output in a quantum system with associated Hilbert space $\kk$ is a completely positive (CP) map $\Lambda : \lh\to\lk$ such that $\tr{\Lambda(T)} = \tr{T}$ for all $T\in\lh$. We denote by $\obs$ the set of all observables and by $\chan$ the collection of all channels. In our definitions of $\obs$ and $\chan$, the Hilbert space $\hh$ is fixed; however, we allow for all possible finite outcome sets $\Omega$ and finite dimensional output Hilbert spaces $\kk$.

An {\em instrument} with outcome set $\Omega$ and output $\lk$ is a collection of CP maps $\Ii = \{\Ii_\omega:\lh\to\lk \mid \omega\in\Omega\}$ such that $\Ii^\chan := \sum_\omega \Ii_\omega$ is a channel; we call it the {\em associated channel} of $\Ii$. We can also define an {\em associated observable} $\Ii^\obs : \Omega\to\lh$, given by $\tr{T\Ii^\obs(\omega)} = \tr{\Ii_\omega(T)}$ for all $T\in\lh$.

An observable $\A$ and a channel $\Lambda$ are {\em compatible} if there exists an instrument $\Ii$ such that $\Ii^\chan = \Lambda$ and $\Ii^\obs = \A$; in this case, we use the shorthand notation $\A\comp\Lambda$. Otherwise, $\A$ and $\Lambda$ are called {\em incompatible}.
Concrete examples of compatible and incompatible pairs can be found in \cite{HeReRyZi18}, where the compatibility relation on certain classes of qubit observables and channels is fully determined. 

For fixed $\A\in\obs$ and $\Lambda\in\chan$, we introduce the following sets associted to the compatibility relation:
\begin{equation}\label{eq:def_sigma_tau_singleton}
\csigma(\A) = \{ \Gamma \in \chan \mid \A \comp \Gamma \}\,, \qquad \ctau(\Lambda) = \{ \B \in \obs \mid \B \comp \Lambda \} \, .
\end{equation}
The main goal of this paper is to study these sets. In the following section, we will extend the previous definitions by replacing $\A$ and $\Lambda$ with collections of observables and channels, respectively. We will show that such natural extensions have a clear operational meaning, and then investigate their properties. 

\subsection{Qualitative noise-disturbance relation}

The sets $\obs$ and $\chan$ have operationally motivated preorders, and the qualitative noise-dis\-tur\-bance relation links these preorders. 
The preorders in question are the post-processing preorders;
for two observables $\A$ and $\B$, we denote $\A \preccurlyeq \B$ if $\A = \mu \circ \B$ for some stochastic matrix (also called stochastic kernel or Markov kernel) $\mu$, where 
\begin{equation*}
(\mu \circ \B)(\omega') = \sum_\omega \mu(\omega',\omega) \B(\omega) \, .
\end{equation*}
Analogously, for two channels $\Lambda$ and $\Gamma$, we denote $\Lambda \preccurlyeq \Gamma$ if $\Lambda = \Theta \circ \Gamma$ for some channel $\Theta$, where $\Theta \circ \Gamma$ is the usual composition of maps.
We say that $\A$ and $\B$ are {\em equivalent} and denote it by $\A \simeq \B$ if both $\A \preccurlyeq \B$ and 
$\B \preccurlyeq \A$ hold. The equivalence relation $\Lambda \simeq \Gamma$ is defined in a similar way. 

If we look at the corresponding equivalence classes, these preorder relations become partial orderings. It is immediate to see that the set \(\chan/\simeq\) has the greatest element, which is the equivalence class of the \emph{identity channel}.
This equivalence class is explicitly described in \cite{NaSe07}. 
The set \(\chan/\simeq\) has also the lowest element, which is the set of all \emph{completely depolarizing channels}, i.e., all channels of the form $\Lambda(T) = \tr{T}\eta$ for some fixed state $\eta$ \cite[Prop.~10]{HeMi17}.  

The partial order structure of \(\obs/\simeq\) is more subtle and it was clarified in \cite{MaMu90a}. 
All trivial observables are equivalent and define the lowest element. 
Here, we recall that a {\em trivial observable} (coin-tossing observable) is any observable of the form $\A(\omega) = p(\omega) \id$ for some probability 
distribution $p: \Omega \to [0,1]$.
On the other hand, there is no greatest element: maximal observables are exactly those whose all nonzero operators are rank-\(1\), and there is infinitely many different equivalence classes of maximal observables.

The preorder structure described above underlies the formulation of the qualitative noise-disturbance relation. 
It translates into the earlier notation as follows.

\begin{theorem}[Theorems 1 and 2 of \cite{HeMi13}]\label{thm:qnd}
\begin{enumerate}[(a)]
\item (Existence of a least disturbing channel for a given observable.) For any observable $\A\in\obs$, let $(V,\kk,\hat{A})$ be a Naimark dilation of $\A$; i.e., $\kk$ is a Hilbert space, $V:\hh\to\kk$ is an isometry and $\hat{\A}:\Omega\to\lk$ is a projection-valued-measure such that $V^*\hat{\A}(\omega)V = \A(\omega)$ for all $ \omega\in\Omega$. 
Then, we have 
\begin{align}\label{eq:a_thm_qnd}
\csigma(\A) = \{ \Lambda \in \chan \mid \Lambda \preccurlyeq \Lambda_\A \}\,,
\end{align}
where the channel $\Lambda_\A : \lh\to\lk$ is defined as
\begin{equation}\label{eq:least_dist}
\Lambda_\A(T) = \sum_\omega \hat{\A}(\omega) V T V^* \hat{\A}(\omega) \, .
\end{equation}
\item (The noise-disturbance trade-off.) For two observables $\A,\B\in\obs$, the following equivalence holds:
\begin{align}\label{eq:b_thm_qnd}
\csigma( \A ) \subseteq \csigma( \B  ) \quad \Leftrightarrow \quad \B \preccurlyeq \A \, .
\end{align}
\end{enumerate}
\end{theorem}
The equivalence class of the channel $\Lambda_\A$ defined in \eqref{eq:least_dist} is the set of all \emph{least disturbing channels} compatible with $\A$. The first part of Theorem~\ref{thm:qnd} can be rephrased by saying that $\csigma(\A)$ is a principal ideal, generated by $\Lambda_\A$. Here, an ideal is meant in the order-theoretic sense.
Combining \eqref{eq:a_thm_qnd} and \eqref{eq:b_thm_qnd} we conclude that
\begin{align}\label{eq:c_thm_qnd}
\B \preccurlyeq \A  \quad \Leftrightarrow \quad \Lambda_\A \preccurlyeq \Lambda_\B \, .
\end{align}

Theorem~\ref{thm:qnd} is about $\sigma_c$ and hence one can ask if something analogous is true for $\tau_c$.
This is not the case, as one observes by inspecting some examples.
Firstly, for every least disturbing channel $\Lambda_\A$, we have 
\begin{equation}\label{eq:ldc-1}
\ctau(\Lambda_\A) = \{\B\in\obs\mid\B \preccurlyeq \A\}
\end{equation}
and, in particular, $\ctau(\Lambda_\A)$ is a principal ideal. 
Indeed, $\B\comp\Lambda_\A$ means that $\Lambda_\A\in\csigma(\B)$, which is equivalent to $\B \preccurlyeq \A$ by combining \eqref{eq:a_thm_qnd} and \eqref{eq:c_thm_qnd}.
For general $\Lambda\in\chan$, however, $\ctau(\Lambda)$ is not a principal ideal. For instance, let \(\Lambda\) be a completely depolarizing channel, i.e., $\Lambda(T) = \tr{T}\eta$ for some fixed state $\eta$. 
We then have \(\ctau(\Lambda)=\obs\), as for any observable $\A$ we can write the instrument $\Ii_\omega(T) = \tr{T\A(\omega)} \eta$ that shows the compatibility of $\A$ and $\Lambda$. Since the set \(\obs\) has inequivalent post-processing maximal elements, \(\csigma(\Lambda)\) is not a principal ideal.

\subsection{Simulability}\label{sec:simulability}

The post-processing relation on observables generalizes to a preorder on the respective power set $2^\obs$, as discussed and used in various ways in \cite{OsGuWiAc17,GuBaCuAc17,FiHeLe18}. 
Namely, suppose $X,X'\subseteq\obs$ are two arbitrary subsets. 
We say that $X'$ is {\em simulable} by $X$ and write $X'\preccurlyeq X$ if for all $\A'\in X'$ there exist $\A_1,\ldots,\A_n\in X$ such that
\begin{equation}\label{eq:sim_obs}
\A' = \sum_i t_i \ \mu_i \circ \A_i 
\end{equation}
for some stochastic matrices $\mu_1,\ldots,\mu_n$ and real numbers $t_1,\ldots,t_n\in [0,1]$ satisfying $\sum_i t_i = 1$. In particular, for singleton sets $\{\A'\}$ and $\{\A\}$, the simulability relation coincides with the post-processing preorder defined earlier, as we have $\{\A'\} \preccurlyeq \{\A\}$ $\Leftrightarrow$ $\A'\preccurlyeq\A$.

Clearly, $X'\subseteq X$ implies $X'\preccurlyeq X$. However, in contrast to the set inclusion relation, the simulability relation is not antisymmetric, hence it constitutes only a preorder on the power set $2^\obs$. 
Also in this case, to get a partial order we need to consider the quotient set $2^\obs / \simeq$ with respect to the equivalence relation $X'\simeq X$ $\Leftrightarrow$ $X'\preccurlyeq X$ and $X\preccurlyeq X'$.

As in \cite{FiHeLe18}, we further introduce the set
$$
\simu_\obs(X) = \{\A\in\obs\mid \{ \A \} \preccurlyeq X\} 
$$
which is the largest subset of $\obs$ that is simulable by $X$. 
As shown in \cite{FiHeLe18},  $\simu_\obs(X)$ is a convex set containing $X$ and $\simu_\obs(\simu_\obs(X)) = \simu_\obs(X)$. 
We also use the shorthand notation $\simu_\obs(\A) \equiv \simu_\obs(\{\A\})$.

We can define simulability for two subsets $Y,Y'\subseteq\chan$ in an analogous way: in \eqref{eq:sim_obs}, it suffices to replace the observables $\A',\A_1,\ldots,\A_n$ with channels $\Lambda'\in Y'$ and $\Lambda_1,\ldots,\Lambda_n\in Y$, and  stochastic matrices $\mu_1,\ldots,\mu_n$ with channels $\Theta_1,\ldots,\Theta_n$.
The definition and properties of $\simu_\chan(Y)$ are similar to $\simu_\obs(X)$.

For the later developments, we record the trivial observation that the statements of Theorem~\ref{thm:qnd}  can be rephrased as
\begin{gather}
\csigma(\A) = \simu_\chan(\Lambda_\A)
\label{eq:thm_qnd_sing_a}\\
\csigma(\A) \subseteq \csigma(\B) \quad \Leftrightarrow \quad \simu_\obs(\B) \subseteq \simu_\obs(\A) \, .\label{eq:thm_qnd_sing_b}
\end{gather}
Finally, \eqref{eq:ldc-1} takes the form
\begin{equation}\label{eq:ldc-2}
\ctau(\Lambda_\A) = \simu_\obs(\A) \, .
\end{equation}
In the following section, we will see how the maps $\csigma$ and $\ctau$ can be naturally generalized and how their properties connect to the simulation maps.

\section{Galois connections and compatibility}\label{sec:galois}

\subsection{General definition of a Galois connection}

In the following, we first recall the basic definitions of Galois connections and closure maps \cite{Ore44,Everett44}.

Let $\mathcal{A}$ and $\mathcal{B}$ be two sets. 
A \emph{Galois connection} between $\mathcal{A}$ and $\mathcal{B}$ is a pair of maps $\sigma:2^{\mathcal{A}} \to 2^{\mathcal{B}}$ and $\tau:2^{\mathcal{B}} \to 2^{\mathcal{A}}$, satisfying the following relations:
\begin{equation}\label{eq:gc1}\tag{GC1}
\begin{gathered}
X' \subseteq X \quad \Rightarrow \quad \sigma(X') \supseteq \sigma(X) \qquad \text{for all $X,X'\subseteq\mathcal{A}$\,,} \\
Y' \subseteq Y \quad \Rightarrow \quad \tau(Y') \supseteq \tau(Y) \qquad \text{for all $Y,Y'\subseteq\mathcal{B}$}
\end{gathered}
\end{equation}
and
\begin{equation}\label{eq:gc2}\tag{GC2}
\begin{gathered}
X \subseteq \tau\sigma(X) \qquad \text{for all $X\subseteq\mathcal{A}$\,,} \\
Y \subseteq \sigma\tau(Y) \qquad \text{for all $Y\subseteq\mathcal{B}$\,.}
\end{gathered}
\end{equation}

Any relation $R$ between the sets $\mathcal{A}$ and $\mathcal{B}$ (i.e., any subset $R\subseteq \mathcal{A}\times\mathcal{B}$) generates an {\em induced Galois connection}.
Namely, by defining
\begin{equation}\label{eq:gcR}
\begin{aligned}
\sigma_R(X) & = \cap_{a \in X} \{ b\in \mathcal{B} : (a,b)\in R \} \\
\tau_R(Y) & = \cap_{b \in Y} \{ a\in \mathcal{A} : (a,b)\in R \}
\end{aligned}
\end{equation}
we obtain maps $\sigma_R$ and $\tau_R$ that satisfy \eqref{eq:gc1}--\eqref{eq:gc2}. 

We further recall that a map $\mathfrak{c}:2^\mathcal{A} \to 2^\mathcal{A}$ is a \emph{closure map} on a set $\mathcal{A}$ if it satisfies the following conditions:
\begin{gather}
X \subseteq \clo(X) \label{eq:cl1} \tag{CL1} \\
\clo(\clo(X)) = \clo(X) \label{eq:cl2} \tag{CL2} \\
X' \subseteq X \quad \Rightarrow \quad \clo(X') \subseteq \clo(X) \label{eq:cl3} \tag{CL3}
\end{gather}
for all $X,X'\subseteq\mathcal{A}$.
A subset $X\subseteq \mathcal{A}$ is called \emph{$\mathfrak{c}$-closed} if $\mathfrak{c}(X)=X$.

The following result is standard and easy to prove \cite[Thm.~2.3.2]{KoppitzDenecke}.

\begin{proposition}
Let $(\sigma,\tau)$ be a Galois connection between sets $\mathcal{A}$ and $\mathcal{B}$.
Then
\begin{enumerate}[(a)]
\item $\sigma\tau\sigma = \sigma$ and $\tau\sigma\tau = \tau$;
\item $\tau\sigma$ and $\sigma\tau$ are closure maps on $\mathcal{A}$ and $\mathcal{B}$, respectively;
\item the $\tau\sigma$-closed sets are all sets of the form $\tau(Y)$ for some $Y\subseteq\mathcal{B}$, and the $\sigma\tau$-closed sets are all sets of the form $\sigma(X)$ for some $X\subseteq\mathcal{A}$.
\end{enumerate}
\end{proposition}

We say that $\tau\sigma$ and $\sigma\tau$ are the closure maps {\em associated} with the Galois connection $(\sigma,\tau)$.

\subsection{Galois connection induced by the compatibility relation}

In the rest of this paper, we are going to investigate the Galois connection induced by the compatibility relation between channels and observabels. 
To do it, we extend the definition of $\csigma$ and $\ctau$ given in \eqref{eq:def_sigma_tau_singleton} from singleton sets to arbitrary subsets $X\subseteq\obs$ and $Y\subseteq\chan$ as follows:
\begin{align*}
\csigma(X) &= \{ \Lambda \in \chan \mid \text{$\A\comp\Lambda$ for every $\A\in X$} \} \, , \\
\ctau(Y) &= \{ \A \in \obs \mid \text{$\A\comp\Lambda$ for every $\Lambda\in Y$} \} \,.
\end{align*}
These maps are then exactly the Galois connection induced by the compatibility relation as done in \eqref{eq:gcR}. 
Therefore, all the previously mentioned general results are valid for $\csigma$ and $\ctau$.
Especially, $\ctau\csigma$ and $\csigma\ctau$ are closure maps.

Our first observation is that the sets $\csigma(X)$ and $\ctau(Y)$ are order-theoretic ideals, as stated in the following simple but useful result.

\begin{proposition}\label{lem:order_Galois}
For all $X,X'\subseteq\obs$ and $Y,Y'\subseteq\chan$ the following implications hold:
\begin{enumerate}[(a)]
\item If $X'\preccurlyeq X \subseteq \ctau(Y)$, then $X'\subseteq \ctau(Y)$. In particular, $\simu_\obs(\ctau(Y)) = \ctau(Y)$.\label{it:a_lem_order_Galois}
\item If $Y'\preccurlyeq Y \subseteq \csigma(X)$, then $Y'\subseteq \csigma(X)$. In particular, $\simu_\chan(\csigma(X)) = \csigma(X)$.\label{it:b_lem_order_Galois}
\end{enumerate}
\end{proposition}

\begin{proof}
If $\Ii$ is an instrument and $\Theta$ is a channel with matching output and input spaces, we can define the new instrument $\Theta\circ\Ii$ given by $(\Theta\circ\Ii)_\omega = \Theta\circ \Ii_\omega$. Similarly, if $\mu$ is a stochastic matrix, we can define the instrument $\mu\circ\Ii$ as $(\mu\circ\Ii)_{\omega'} = \sum_\omega \mu(\omega',\omega) \Ii_\omega$. It is easy to check that
\begin{align*}
(\Theta\circ\Ii)^\chan & = \Theta\circ\Ii^\chan\,, & (\Theta\circ\Ii)^\obs & = \Ii^\obs\,, \\
(\mu\circ\Ii)^\chan & = \Ii^\chan\,, & (\mu\circ\Ii)^\obs & = \mu\circ\Ii^\obs \,.
\end{align*}
We use the two relations in the second row to prove \eqref{it:a_prop_order_Galois}. 
The proof of \eqref{it:b_prop_order_Galois} is similar.

Assume $X'\preccurlyeq X \subseteq \ctau(Y)$, and let $\A'\in X'$. 
Then $\A'$ can be expressed as in \eqref{eq:sim_obs} for some choice of $\A_1,\ldots\A_n\in X$, stochastic matrices $\mu_1,\ldots,\mu_n$ and real numbers $t_1,\ldots,t_n\in [0,1]$ satisfying $\sum_i t_i = 1$. 
For any $\Lambda\in Y$, fix instruments $\Ii_1,\ldots\Ii_n$ such that $\Ii_i^\chan = \Lambda$ and $\Ii_i^\obs = \A_i$ for all $i$; moreover, let $\Ii' = \sum_i t_i \mu_i\circ\Ii_i$. 
Then $\Ii^{\prime\,\chan} = \Lambda$ and $\Ii^{\prime\,\obs} = \A'$. 
We thus conclude that $\A'\in \ctau(Y)$, hence $X'\subseteq \ctau(Y)$. In particular, by choosing $X=\ctau(Y)$ and $X'=\simu_\obs(\ctau(Y))$, we find the inclusion $\simu_\obs(\ctau(Y))\subseteq\ctau(Y)$. The reverse inclusion is trivial, and therefore $\simu_\obs(\ctau(Y)) = \ctau(Y)$.
\end{proof}

A first consequence of Proposition \ref{lem:order_Galois} is that conditions \eqref{eq:gc1}-\eqref{eq:gc2} hold for the maps $\csigma$ and $\ctau$ also if we replace the partial order $\subseteq$ with the simulability preorder $\preccurlyeq$. 
Indeed, we even have a bit stronger fact, as shown by the next result.

\begin{proposition}\label{prop:order_Galois}
For all $X,X'\subseteq\obs$ and $Y,Y'\subseteq\chan$ the following implications hold:
\begin{enumerate}[(a)]
\item If $X'\preccurlyeq X$, then $\csigma(X')\supseteq\csigma(X)$. In particular, $\csigma(X) = \csigma(\simu_\obs(X))$. \label{it:a_prop_order_Galois}
\item If $Y'\preccurlyeq Y$, then $\ctau(Y')\supseteq\ctau(Y)$. In particular, $\ctau(Y) = \ctau(\simu_\chan(Y))$.\label{it:b_prop_order_Galois}
\end{enumerate}
\end{proposition}

\begin{proof}
Suppose $X'\preccurlyeq X$. 
Since $X\subseteq\ctau\csigma(X)$, we have $X'\subseteq\ctau\csigma(X)$ by Proposition \ref{lem:order_Galois}. 
Then, $\csigma(X')\supseteq\csigma\ctau\csigma(X) = \csigma(X)$, as claimed in \eqref{it:a_prop_order_Galois}. 
In the particular case $X'=\simu_\obs(X)$, we have both $X'\preccurlyeq X$ and $X\preccurlyeq X'$, hence the equality $\csigma(X) = \csigma(\simu_\obs(X))$ holds. 
The proof of \eqref{it:b_prop_order_Galois} is similar.
\end{proof}

Next, we study the interplay between compatibility closure maps and simulability. 
To this aim, we recall that also $\simu_\obs$ is a closure map \cite{FiHeLe18}, and it is easy to observe that the same is true for $\simu_\chan$. 
As a consequence of Propositions \ref{lem:order_Galois} and \ref{prop:order_Galois}, we see that these closure maps have the following relation with the closure maps associated with the Galois connection $(\csigma,\ctau)$.

\begin{proposition}\label{prop:simu-Galois}
We have 
\begin{equation}\label{eq:sg-1}
\simu_\obs(X)\subseteq\ctau\csigma(X) = \ctau\csigma(\simu_\obs(X))
\end{equation}
and
\begin{equation}\label{eq:sg-2}
\simu_\chan(Y)\subseteq\csigma\ctau(Y) = \csigma\ctau(\simu_\chan(Y))
\end{equation}
for all $X\subseteq\obs$ and $Y\subseteq\chan$.
\end{proposition}

\begin{proof}
We prove only \eqref{eq:sg-1}, the proof of \eqref{eq:sg-2} being similar.
Since $\simu_\obs(X)\preccurlyeq X \subseteq \ctau\csigma(X)$, the inclusion $\simu_\obs(X)\subseteq \ctau\csigma(X)$ follows from Proposition \ref{lem:order_Galois}\eqref{it:a_lem_order_Galois}. 
On the other hand, the equality $\ctau\csigma(X) = \ctau\csigma(\simu_\obs(X))$ is a consequence of Proposition \ref{prop:order_Galois}\eqref{it:a_prop_order_Galois}.
\end{proof}

\section{Leak of information}\label{sec:leak}

In the previous section, we observed that the simulation closure map $\simu_\obs$ is related to the closure map $\tau_c \sigma_c$. Here, we describe the operational meaning of the latter closure map.

Let $\Lambda: \mathcal{L}(\mathcal{H}) \to \mathcal{L}( \mathcal{K})$ be a channel. 
Then one can construct a quartet $(\mathcal{V}_1, \mathcal{V}_2, U, | \eta\rangle)$, where $\mathcal{V}_1$ and $\mathcal{V}_2$ 
are Hilbert spaces, 
$U$ is a unitary operator from $\hh \otimes \mathcal{V}_1$
to $\kk \otimes \mathcal{V}_2$ and 
 $\ket{\eta}$ is a normalized vector of $\mathcal{V}_1$, in a way that 
\begin{eqnarray}\label{eq:dilation}
\Lambda(T) = {\rm tr}_{\mathcal{V}_2} [U (T\otimes \kb{\eta}{\eta}) U^*]
\label{realization}
\end{eqnarray}
for any $T\in\lh$ (see e.g.~\cite{Ozawa84,HeMiRe14}). In the last formula, ${\rm tr}_{\mathcal{V}_2} : \mathcal{L}(\kk\otimes\mathcal{V}_2) \to \lk$ denotes the partial trace over the 
$\mathcal{V}_2$-system.
This quartet can be interpreted as a physical realization of the channel $\Lambda$. Indeed, we see from \eqref{eq:dilation} that $\Lambda$ is implemented by introducing an auxiliary $\mathcal{V}_1$-system (apparatus) prepared in the initial state $\kb{\eta}{\eta}$, then making the system and the apparatus interact by means of the unitary evolution $U$, and finally discarding the $\mathcal{V}_2$-subsystem from the resulting compound state.

For each realization of the channel $\Lambda$, an observable 
on $\mathcal{V}_2$ defines a measurement process.  
More precisely, an observable $\F:\Omega\to\mathcal{L}(\mathcal{V}_2)$ defines an instrument $\Ii = \{\Ii_{\omega} : \lh\to\lk \mid \omega\in\Omega\}$ by setting
\begin{eqnarray*}
\Ii_{\omega}(T) = {\rm tr}_{\mathcal{V}_2}
 [U (T\otimes \kb{\eta}{\eta}) U^* (\id\otimes \F(\omega))]\,.
\end{eqnarray*}
Such an observable $\F$ is called a {\em pointer observable}; we measure it on the apparatus after the interaction in order to extract information on the system.

The instrument $\Ii$ describes a measurement of the observable
$$
\A(\omega) = V^* (\id\otimes \F(\omega)) V
$$
on the $\hh$-system, where $V\ket{\psi} = U \ket{\psi}\otimes\ket{\eta}$; therefore, we may call $\A$ the observable \emph{induced} by $\F$. By the very definition, this induced observable and the channel $\Lambda$ are compatible. 

Furthermore, according to Radon-Nikodym theorem \cite{Arveson69,Raginsky03}, one can find that \emph{the set of all the observables compatible with 
$\Lambda$ (i.e., $\ctau(\Lambda)$) 
coincides with the set of all the induced observables 
obtained by all the possible choices of pointer observables}.  
Clearly, this set does not depend on the realization $(\mathcal{V}_1, \mathcal{V}_2, U, | \eta\rangle)$.

Now,  suppose that we have a realization of a channel $\Lambda$ which is compatible with an observable $\A$. Then surely $\A \in \ctau(\Lambda)$ holds. 
Now the question is if there is some other induced observable which can be obtained
by choosing a different pointer observable for any $\Lambda$ compatible with $\A$.
This subset of observables, which we call \emph{leak of information} for $\A$, 
is represented by $\ctau \csigma (\A)$.
It is hence given by one of the closure maps discussed in Section \ref{sec:galois}.

We can generalize this notion to a subset $X$ of observables. 
The question is then: what is the set of observables each of which can be measured by suitably choosing a pointer observable for any $\Lambda$ compatible with every $\A \in X$? 
The seeked set is clearly equal to $\ctau\csigma(X)$.
Motivated by this physical interpretation, we denote
\begin{equation*}
\leak = \ctau\csigma
\end{equation*}
and call this map the \emph{leak closure}.
We observe that Proposition \ref{prop:simu-Galois} implies the inclusion $\simu_\obs(X) \subseteq \leak(X)$ for all $X\subseteq\obs$.

Although $\leak(X)$ for general $X\subset\obs$ can be difficult to be determined, for certain sets it has a neat form. This is the content of the next result. 

\begin{theorem}\label{thm:leak-sim}
Let $X\subset\obs$ be a set having a greatest element. That is, there exists an element $\A\in X$ such that $\B \preccurlyeq \A$ holds for any $\B\in X$. 
Then
\begin{eqnarray*}
\leak(X) = \simu_\obs(\A) \,.
\end{eqnarray*}
\end{theorem}

\begin{proof}
Under the conditions of the theorem, we have $\{\B\}\preccurlyeq X$ $\Leftrightarrow$ $\{\B\}\preccurlyeq\{\A\}$, hence $\simu_\obs(X) = \simu_\obs(\A)$. 
We conclude that
$$
\csigma(X) = \csigma(\A) = \simu_\chan(\Lambda_\A)
$$
where the first equality follows from Proposition \ref{prop:order_Galois}\eqref{it:a_prop_order_Galois} and the second one is \eqref{eq:thm_qnd_sing_a}. 
Combining this with Proposition \ref{prop:order_Galois}\eqref{it:b_prop_order_Galois} and \eqref{eq:ldc-2}, we get
$$
\leak(X)=\ctau(\csigma(X)) = \ctau(\simu_\chan(\Lambda_\A)) = \ctau(\Lambda_\A) =  \simu_\obs(\A) \, ,
$$
which proves the theorem.
\end{proof}

A subset $X\subset\obs$ as in the above theorem can be regarded as a classical set since it admits a most informative observable $\A$. 
We note that the theorem specializes the general inclusion $\simu_\obs(X) \subseteq \leak(X)$ to the equality $\simu_\obs(X)=\leak(X)$ whenever $X = \{\A\}$ is a singleton set. 
In the following two examples, however, we demonstrate that the equality $\simu_\obs(X)=\leak(X)$ does not always hold. 

\begin{example}\label{ex:qubit}
Let us consider $\hh= \Cb^2$.
We fix the three Pauli matrices $\vsigma = (\sigma_1,\sigma_2,\sigma_3)$, and define two sharp qubit observables 
$$
\A(\pm) = \half (\id \pm \sigma_1) \, , \quad \B(\pm) = \half (\id \pm \sigma_2) \, .
$$
For these observables, we claim that
\begin{equation}\label{eq:qubit-ex}
\leak(\{\A, \B\})= \obs \neq \simu_\obs(\{\A, \B\}) \,.
\end{equation}
In order to prove the equality, it suffices to show that 
$\csigma(\{\A, \B\})$ consists of completely depolarizing channels. 
Any $\Lambda \in \csigma(\{\A, \B\})= \csigma(\A) \cap \csigma(\B)$
is written as 
$$
\Lambda(T) = \sum_{\omega} \mbox{tr}[T \A(\omega)]\eta(\omega)
= \sum_{\omega} \mbox{tr}[T \B(\omega)] \xi(\omega)
$$
with some states $\eta(\omega)$ and $\xi(\omega)$ for $\omega = \pm$ 
\cite[Cor~1]{HeWo10}. Putting 
$T= \A(\omega)$ for $\omega=\pm$, we conclude
$$
\eta(+) = \eta(-)= \half(\xi(+) + \xi(-)) \, .
$$
Thus we observe that $\Lambda$ is a completely depolarizing channel. 
Further, the inequality in \eqref{eq:qubit-ex} follows since the set $\simu_\obs(\{\A, \B\})$ is contained in the linear span of the operators $\id$, $\sigma_1$ and $\sigma_2$; hence e.g. the observable $\C(\pm) = (\id \pm \sigma_3)/2$ can not be an element of $\simu_\obs(\{\A, \B\})$.
\end{example}

We recall that two observables $\A$ and $\B$ are called {\em jointly measurable} if there exists a third observable $\G$ such that $\A\preccurlyeq\G$ and $\B\preccurlyeq\G$; otherwise $\A$ and $\B$ are incompatible.
In Example \ref{ex:qubit}, the two observables $\A$ and $\B$ are incompatible.
On the other hand, all the observables in the set $X$ of Theorem \ref{thm:leak-sim} are jointly measurable, as they are post-processings of the greatest element $\A$. One may then wonder if joint measurability is a sufficient condition for the equality $\leak(X) = \simu_\obs(X)$. This is not the case, as the next slightly more elaborate example shows.

\begin{example}
Let us consider $\hh= \Cb^3$. 
We fix an orthonormal basis $\{|n\rangle\}_{n=1,2,3}$ and define an observable $\Eo$ as 
 $\Eo(n)= |n\rangle \langle n|$, $n=1,2,3$. 
We then introduce two other observables $\A$ and $\B$, given as $\A(1)= \Eo(1)$, $\A(2)=\Eo(2)+\Eo(3)$ and $\B(1)=\Eo(1)+\Eo(2)$, $\B(2)=\Eo(3)$. We claim that
\begin{equation}\label{eq:chain_leak}
\leak(\{\A, \B\}) = \simu_\obs (\Eo) \neq \simu_\obs (\{\A, \B\}) \,.
\end{equation}
In order to prove the left equality, first of all we observe that
\begin{equation}\label{eq:csigma(A,B)}
\csigma(\{\A, \B\}) = \simu_\chan(\Lambda_\Eo) \,.
\end{equation}
Indeed, since $\{\A, \B\} \preccurlyeq \{\Eo\}$, we have 
$$
\csigma(\{\A, \B\}) \supseteq \csigma(\Eo) = \simu_\chan(\Lambda_\Eo)
$$
by Proposition \ref{prop:order_Galois}\eqref{it:a_prop_order_Galois} and \eqref{eq:thm_qnd_sing_a}. On the other hand, if $\Lambda\in\csigma(\{\A, \B\}) = \csigma(\A)\cap\csigma(\B)$, then necessarily $\Lambda\preccurlyeq\Lambda_\A$ and $\Lambda\preccurlyeq\Lambda_\B$ by \eqref{eq:a_thm_qnd}. Using the trivial Naimark dilations of the two projection-valued-measures $\A$ and $\B$, formula \eqref{eq:least_dist} yields
\begin{align*}
\Lambda_\A(T) & = \Eo(1)T\Eo(1) + (\Eo(2)+\Eo(3)) T (\Eo(2)+\Eo(3))\,, \\
\Lambda_\B(T) & = (\Eo(1)+\Eo(2))T(\Eo(1)+\Eo(2)) + \Eo(3)T\Eo(3) \,.
\end{align*}
We see that
$$
\Lambda_\A(\kb{1}{i}) = \Lambda_\A(\kb{i}{1}) = \Lambda_\B(\kb{3}{j}) = \Lambda_\B(\kb{j}{3}) = 0 \quad \text{if $i\neq 1$, $j\neq 3$} \,,
$$
hence the same equalities must hold with the channel $\Lambda$ replacing $\Lambda_\A$ and $\Lambda_\B$. It follows that $\Lambda$ is the measure-and-prepare channel
$$
\Lambda(T) = \sum_{i=1}^3 \tr{T\Eo(i)} \eta_i \,,
$$
where $\eta_1$, $\eta_2$ and $\eta_3$ are three fixed states. We clearly have $\Eo\comp\Lambda$, hence $\Lambda\in\simu_\chan(\Lambda_\Eo)$ by \eqref{eq:thm_qnd_sing_a}. This proves the inclusion $\csigma(\{\A, \B\})\subseteq\simu_\chan(\Lambda_\Eo)$, and thus completes the proof of \eqref{eq:csigma(A,B)}. Applying $\ctau$ to both sides of \eqref{eq:csigma(A,B)} and using Proposition \ref{prop:order_Galois}\eqref{it:b_prop_order_Galois}, we get the left equality in \eqref{eq:chain_leak}.
Finally, we have $\Eo \notin \simu_\obs(\{\A, \B\})$ since ${\rm rank}\,\Eo(i) = 1$ for all $i=1,2,3$ while ${\rm rank}\,\A(2) = {\rm rank}\,\B(1) = 2$. This proves the right inequality in \eqref{eq:chain_leak}.
\end{example}

\section{Joint measurement closure map}\label{sec:joint}

In the previous section, we have introduced $\leak$ as the closure map on $\obs$ given by the Galois connection $(\csigma,\ctau)$. 
We have also discussed the physical interpretation of $\leak(X)$ for a subset of observables $X\subseteq\obs$, and we have observed the inclusion $\simu_\obs(X) \subseteq \leak(X)$. 
We have also seen that $\simu_\obs(X) = \leak(X)$ holds in some cases but not in all.

In this section, we introduce a third closure map on $\obs$ and describe its relation to $\leak$.
The \emph{joint measurement closure map} $\joint$ is the closure map that is determined by the joint measurability relation via Galois connection. 
In details, for a subset $X\subseteq\obs$, we denote by $J(X)$ the set of all observables $\B$ that are jointly measurable with every $\A\in X$. 
That is, $J: 2^{\obs} \to 2^{\obs}$ is defined by
$$
J(X) = \{\B\in\obs\mid\text{$\B$ is jointly measurable with every $\A\in X$}\}\,.
$$
Then, $(J,J)$ is the Galois connection induced by the joint measurability relation between observables as in \eqref{eq:gcR}. We denote by
\begin{equation*}
\joint = J^2
\end{equation*}
the associated closure map. The proofs of Propositions \ref{lem:order_Galois}, \ref{prop:order_Galois} and \ref{prop:simu-Galois} can be straightforwardly rewritten also for the Galois connection $(J,J)$.
In particular, we have
\begin{equation*}
J(X) = J(\simu_\obs(X)) \,,\qquad\quad \simu_\obs(X)\subseteq\joint(X) = \joint(\simu_\obs(X)) \,.
\end{equation*}

The following theorem establishes the relation between the closure maps $\joint$ and $\leak$.

\begin{theorem}\label{thm:leak-joint}
For any $X\subseteq\obs$, 
$$
\leak(X) \subseteq \joint(X) \, .
$$
\end{theorem} 

\begin{proof}
According to \cite{HeMi17}, for any observable $\A$ with outcomes in $\Omega$, we can define a measure-and-prepare channel $\Gamma_\A : \lh \to \mathcal{L}(\ell^2(\Omega))$, given as
$$
\Gamma_\A(\varrho) = \sum_\omega \tr{\varrho\A(\omega)} \kb{\delta_\omega}{\delta_\omega} \,.
$$
Here, $\ell^2(\Omega)$ is the Hilbert space of all complex valued functions on $\Omega$ endowed with the scalar product $\ip{f}{g} = \sum_\omega \overline{f(\omega)} g(\omega)$, and $\{\delta_\omega\}_{\omega\in\Omega}$ is the orthonormal basis of $\ell^2(\Omega)$ made up of all delta functions. By \cite[Prop.~7]{HeMi17}, we have the equivalences
$$
\A\in J(X) \quad\Leftrightarrow\quad \Gamma_A\in\csigma(X) \quad\Leftrightarrow\quad \A\in\ctau(\{\Gamma_\B\mid\B\in X\}) \,.
$$
Hence,
\begin{align*}
& \csigma(X) \supseteq \{\Gamma_\A\mid\A\in J(X)\} \\
& \qquad\Rightarrow\quad \ctau(\csigma(X)) \subseteq \ctau(\{\Gamma_\A\mid\A\in J(X)\}) = J(J(X)) \,,
\end{align*}
which is the claim.
\end{proof}

In the following example, we demonstrate that Theorem \ref{thm:leak-joint} can be used to obtain information about $\leak(X)$.

\begin{example}
This example is related to \cite{Busch86} where the compatibility of two unbiased qubit observables was characterized.
Let $\hh= \Cb^2$. 
An unbiased dichotomic observable $\A_\va$ is described 
as 
$$
\A_\va(\pm) = \half (\id\pm \va \cdot \vsigma) \, , 
$$
where $\va\in\Rb^3$ satisfies $\no{\va} \leq 1$; here, the value of $\no{\va}$ is the sharpness parameter. 
For each $\lambda\in [0,1]$, we introduce the set of observables
$$
\mathcal{A}_{\lambda} = \left\{\A_\va \mid \no{\va} \leq \lambda \right\} \,.
$$
In particular, $\mathcal{A}_1$ is the set of all unbiased dichotomic observables. Clearly, $\mathcal{A}_{\lambda} \subseteq \mathcal{A}_{\lambda'}$ if and only if $\lambda \leq \lambda'$. 
As shown in \cite[Cor.~4.6]{Busch86}, we have
$$
J(\mathcal{A}_{\lambda}) \cap \mathcal{A}_1 = \mathcal{A}_{\sqrt{1-\lambda^2}} \,.
$$
In particular, $J(\mathcal{A}_{\lambda}) \supseteq \mathcal{A}_{\sqrt{1-\lambda^2}}$ holds. 
Thus we obtain 
\begin{equation*}
J(J(\mathcal{A}_{\lambda})) \subseteq J(\mathcal{A}_{\sqrt{1-\lambda^2}}) \,.
\end{equation*}
and then
\begin{equation*}
\joint(\mathcal{A}_{\lambda})
\cap \mathcal{A}_1 
\subseteq \mathcal{A}_{\lambda} \,. 
\end{equation*}
As $\mathcal{A}_{\lambda} \subseteq \leak (\mathcal{A}_{\lambda})$ also holds, due to Theorem \ref{thm:leak-joint} we conclude
\begin{equation*}
\leak(\mathcal{A}_{\lambda}) \cap \mathcal{A}_1 = \joint(\mathcal{A}_{\lambda}) \cap \mathcal{A}_1 = \mathcal{A}_{\lambda}. 
\end{equation*}
Therefore, for any observable $\A_\va$ with $\no{\va} > \lambda$, there exists a channel (respectively, an observable) compatible with all observables in 
$\mathcal{A}_{\lambda}$ such that it is incompatible with $\A_\va$.
\end{example}

\section{Discussion}\label{sec:discussion}

The mathematical formulation of the qualitative noise-disturbance relation roots to the compatibility of observables and channels. The relation fits to the general framework of Galois connections, which led us to introduce the closure map $\leak$ interpreted as the leak of information. This closure map is bounded by other closures as, for each $X \subseteq \mathcal{O}$, 
\begin{eqnarray*}
\simu_{\obs}(X) \subseteq \leak(X) \subseteq \joint(X),  
\end{eqnarray*} 
where $\simu_{\obs}(X)$ and $\joint(X)$ are defined without referring to $\chan$.  
We hope that we have been able to demonstrate that the noise-disturbance relation is a rich topic and there are still many aspects that have not yet been fully explored. 

Paul was one of the pioneers of investigating the mathematical structure and operational properties of quantum measurements.   
His research articles on this topic and three co-authored books \cite{QTM96,OQP97,QM16} serve as a starting point for anyone who wishes to delve into this subject.
We greatly miss him; he was a very generous person who was always open to new ideas and supported us as a mentor and as a friend.

\subsection*{Acknowledgements}
This work was supported by JSPS KAKENHI Grant Number 15K04998. 
 TH  acknowledges financial support from the Academy of Finland via the Centre of Excellence program (Project no.312058) as well as Project no.  287750.

\end{document}